\documentclass[11pt,runningheads]{llncs}
\usepackage[T1]{fontenc}
\usepackage{graphicx}
\usepackage{url}

\usepackage{stmaryrd,euscript,float,bbold,amssymb,amsmath,tikz,hyperref,bookmark}
\usepackage{algorithm}
\usepackage[noend]{algpseudocode}
\usepackage{color}

\newtheorem{prop}{Property}
\begin{document}
\title{Distributed computation of temporal twins in periodic undirected time-varying graphs}
\titlerunning{Distributed computation of temporal twins in periodic undirected TVGs}
\author{Lina Azerouk\inst{1}
\and Binh-Minh Bui-Xuan\inst{1}
\and Camille Palisoc\inst{1}
\and Maria Potop-Butucaru\inst{1}
\and Massinissa Tighilt\inst{1}}

\authorrunning{L. Azerouk et al.}
\institute{LIP6 (CNRS -- Sorbonne Universit\'e), \email{[lina.azerouk,buixuan,camille.palisoc,maria.potop-butucaru,massinissa.tighilt]@lip6.fr}}
\maketitle
\begin{abstract}
Twin nodes in a static network capture the idea of being substitutes for each other for maintaining paths of the same length anywhere in the network. In dynamic networks, we model twin nodes over a time-bounded interval, noted $(\Delta,d)$-twins, as follows. A periodic undirected time-varying graph $\mathcal G=(G_t)_{t\in\mathbb N}$ of period $p$ is an infinite sequence of static graphs where $G_t=G_{t+p}$ for every $t\in\mathbb N$. For $\Delta$ and $d$ two integers, two distinct nodes $u$ and $v$ in $\mathcal G$ are $(\Delta,d)$-twins if, starting at some instant, the outside neighbourhoods of $u$ and $v$ has non-empty intersection and differ by at most $d$ elements for $\Delta$ consecutive instants. In particular when $d=0$, $u$ and $v$ can act during the $\Delta$ instants as substitutes for each other in order to maintain journeys of the same length in time-varying graph $\mathcal G$.
We propose a distributed deterministic algorithm enabling each node to enumerate its $(\Delta,d)$-twins in $2p$ rounds, using messages of size $O(\delta_\mathcal G\log n)$, where $n$ is the total number of nodes and $\delta_\mathcal G$ is the maximum degree of the graphs $G_t$'s.
Moreover, using randomized techniques borrowed from distributed hash function sampling, we reduce the message size down to $O(\log n)$ w.h.p.

\end{abstract}
%
\section{Introduction}
\label{sec:intro}
Maintaining connectivity is an important topic for dynamic networks.
For instance, in the $1$-interval-connected network model~\cite{OW05,KLO10,LV22}, 
if moreover the graph is periodic and each node's local view can be updated entirely every round on contact of its neighbours, then it is proven that a temporally optimal broadcast tree can be maintained in real time~\cite{CFMS14}.
This is important for connectivity because of the real time updates, namely nodes do not need to process a posteriori the pieces of information they received about their neighbourhood.
Here, the $1$-interval-connected model allows for a fixed set of nodes to communicate through a temporal edge set changing at discrete time units called rounds.
The structure as a whole is called a time-varying graph~\cite{CFQS12}, where at every round the graph is also usually assumed to be connected.
\begin{figure}
   \centerline{\includegraphics[scale=0.2]{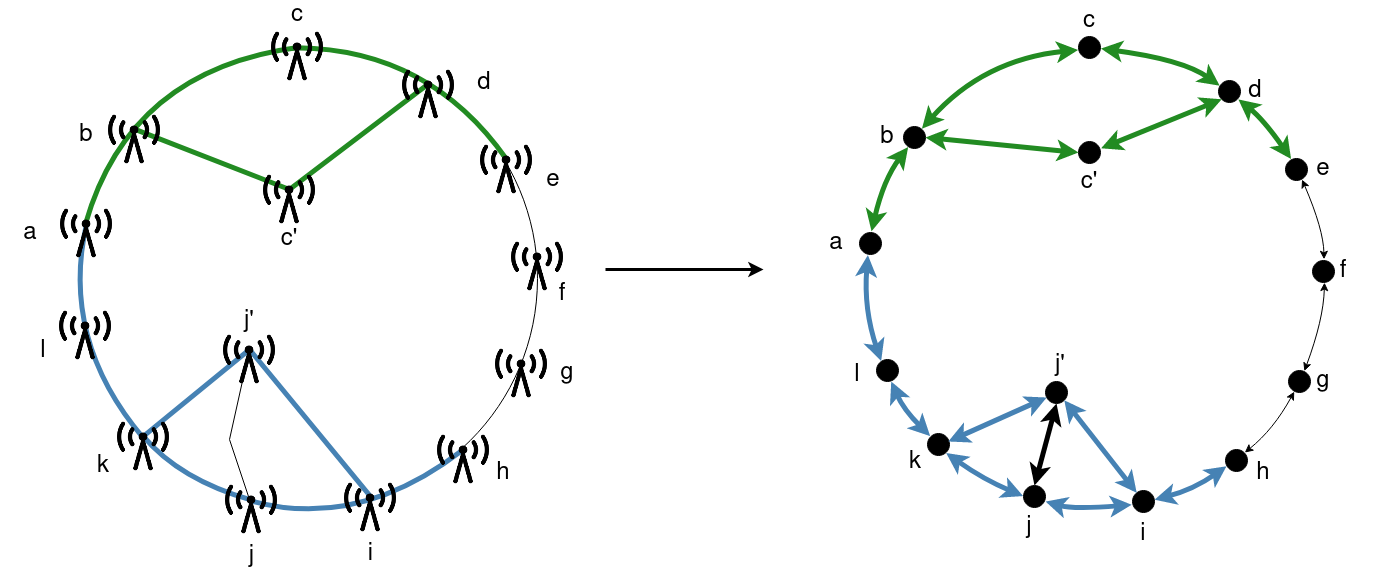}}
    \caption{Network modelled as an undirected graph. Arrows are in double direction to stress that communications are allowed both ways around. The two pairs of twins in the graph are $\{c,c'\}$ and $\{j,j'\}$. A packet traveling between node $a$ and $e$ can choose $(a, b, c, d, e)$ as shortest path. If node $c$ switches off or if it is faulty, the path of same length $(a,b,c',d,e)$ can still be used.
    Between $a$ and $h$, the two paths $(a, l, k, j, i, h)$ and $(a, l, k, j', i, h)$ are shortest, however, $(a,l,k,j,j',i,h)$ is not a shortest path.}
    \label{fig_network}
\end{figure}

In our study, we consider all the above assumptions, except for that we do not require the time-varying graph to be connected at every round.
This is convenient to model adhoc networks with both dense and sparse areas with predictable periodicity such as satellite networks.
We would like to capture by the notion of twin nodes the possibility of backup routes or of a resource saving strategy during the times when some satellites are called twins.
This is why we aim at a formalism where the connectivity in the rest of the network are maintained by journeys of the same length even though one of the twin satellites switches off for sleeping or maintenance during some time window of a given desired length, \textit{cf.} Fig.~\ref{fig_network}.

In a static graph, such a concept can be formalised following the model of $\epsilon$-modules~\cite{HMZ20} and $d$-contractions~\cite{BKTW20}: given an integer $d$, two nodes $u\neq v$ are $d$-twins if their outside neighbourhoods have non-empty intersection while the symmetric difference has size smaller than $d$, that is, the sum of $|N(u)\setminus N(v)\setminus\{u,v\}|$ and $|N(v)\setminus N(u)\setminus\{u,v\}|$ is at most $d$.
We need the former assumption of non-empty intersection for obtaining efficient computational results.
At the same time, isolated nodes or nodes without common neighbours are bad models of backup routes as in the above example of network of satellites.
Hence, we discard them from our definition of $d$-twins which requires nodes to have at least one common neighbour.
From this point, our extension to the dynamic case will follow that of $\Delta$-twins~\cite{BHM20}: given two integers $\Delta$ and $d$, two nodes $u\neq v$ are $(\Delta,d)$-twins in a time-varying graph $\mathcal G=(G_t)_{t\in\mathbb N}$ if there exists $t_0\in\mathbb N$ such that $u$ and $v$ are $d$-twins in every graph $G_t$, for $t_0\leq t<t_0+\Delta$.

\medskip

\textbf{Our contribution:}

We give a distributed deterministic algorithm for every node in a periodic undirected time-varying graph of period $p$ to compute its $(\Delta,d)$-twins after $2p$ rounds, using messages of size $O(\delta_\mathcal G\log n)$, where $\delta_\mathcal G$ is the maximum degree of the graphs $G_t$'s.
When randomized by sampling techniques from~\cite{HNT22}, our messages can be reduced to $O(\log n)$ w.h.p.

The main idea in our algorithm can be divided into two steps.
In the first step, we show how to solve the problem of finding $d$-twins in a static graph $G$ after $2$ rounds.
Roughly, while in $2$ rounds every node can receive messages sent by its $2$-neighbourhood, the receiver node need to detect those sent by its $d$-twins.
In particular, note that while a path over $3$ vertices admits two $0$-twins, a path over $n\geq 4$ vertices has none (because for every node, one vertex from its $3$-neighbourhood breaks down the $0$-twin definition).
In order to avoid examining messages from $k$-neighbourhoods with $k\geq 3$, we exhibit a set property involving the number of paths of length $2$ which allows to detect $d$-twins.
The main idea is to use inclusion-exclusion properties on the neighbourhoods of nodes in order to restrict the message size to $O(\delta_G\log n)$.
Furthermore, by exploiting very simple twist inspired from distributed hash function sampling~\cite{HNT22}, we reduce the message size down to $O(\log n)$ w.h.p.

In a second step, we address the dynamic case.
The main idea here is for the receiver nodes to store information while waiting for the time-varying graph to repeat its edges, using periodicity.
At the same time, a sender node must be detected $\Delta$ consecutive times as a $d$-twin in order to be detected as $(\Delta,d)$-twin.
A special attention is needed for all computations to end after $2p$ rounds (and not $2p+\Delta$).

Eventually, we remark that our algorithm requires the time-varying graph to be undirected, meaning communication on an edge $uv$ must be allowed both way around between $u$ and $v$.
This is because of the step in the static setting where every node attempts to count the number of paths of length $2$ linking itself to every of its $d$-twins.

The paper is organised as follows.
The formal framework of periodic time-varying graphs and $(\Delta,d)$-twins is defined in Section~\ref{sec:def}.
In Section~\ref{sec:disco} we briefly present a property of twins in static graphs, before using it to show a distributed algorithm for every node to compute its $(\Delta,d)$-twins and prove its correctness.
Section~\ref{sec:hash} is devoted to reducing the message size used by our algorithm to $O(\log n)$ w.h.p.
We conclude the paper in Section~\ref{sec:conclu} along with some open remarks for future works.

\section{Model and Problem definition}
\label{sec:def}
We consider distributed systems which are fault-free, message-passing, synchronous, with a unique ID for each process.
This will be formalised as a time-varying graph~\cite{CFQS12} defined by
$\mathcal{G} = (V, E, T, \rho)$, where
$V$ is a finite set of nodes representing processes,
$E\subseteq\binom{V}{2}$ is the underlying set of all possible edges and $\rho : E \times T \to \{0,1\}$ is the presence function defining whether an edge exists in a given round.
Temporal edges are those bound to a specific time instant: $E_t = \{e \in E : \rho(e, t) = 1\}$.
Alternatively, a time-varying graph can also be seen as a sequence of static graphs $\mathcal G=(G_t)_{t\in T}$, where $G_t=(V,E_t)$.

Inspired by the models of dynamic networks in~\cite{CFQS12,KLO10,LV22,OW05}, we suppose furthermore that the time-varying graph is characterized by a periodicity, non-anonymity, synchronicity, and local knowledge.
However, we do not require it to be connected at every time instant.
\begin{itemize}
\item The \textbf{periodicity} means the existence of a specific value $p$ such that, every edge present in the graph at a particular time $t$ has its presence in the graph repeated periodically at $t+p$.
Because of periodicity, we restrict our domain of study down to the cyclic group of $p$ elements and abusively refer to it as $T=\mathbb Z/{p\mathbb Z}=\{0,1,\dots,p-1\}=\llbracket0,p-1\rrbracket$, and subsequently restrict the input time-varying graph to the first $p$ instants, $\mathcal G=(G_t)_{0\leq t<p}$. 
\item The \textbf{non-anonymity} of the graph means that each node $v$ has a unique $ID \in\mathbb N$, we assume that each $ID$ is represented using at most $O(\log n)$ bits, where $n = |V|$.
\item The \textbf{synchronicity} within the graph implies the presence of a global clock that ticks regularly for each $t\in T$, beginning with $t=0$. During each tick of this clock, every node $v$ in the graph performs three sequential actions: sending a message to each neighbouring node, receiving the message from each neighbour, and executing computational tasks. Each occurrence of the clock ticking is called \textit{round}.
\item By \textbf{local knowledge}, we assume that each node $v$ in the graph has the knowledge of the number of its neighbours at any given time $t \in T$, but not their $IDs$. The set of neighbours of node $v$ at time $t$ is denoted as $N_t(v)$.
\end{itemize}
\begin{remark}
If the system does not allow for local knowledge, nodes can still pass their ID to all neighbours at every round $t$.
Then, at round $t+1$, the system satisfies the retro-property of local knowledge of every past round.
  In other words, it is possible to equip any periodic time-varying graph of period $p$ with the local knowledge property after a preprocessing using $p$ rounds and messages of size $O(\log n)$.
\end{remark}
\begin{figure}[t!]
    \centerline{\includegraphics[scale=0.18]{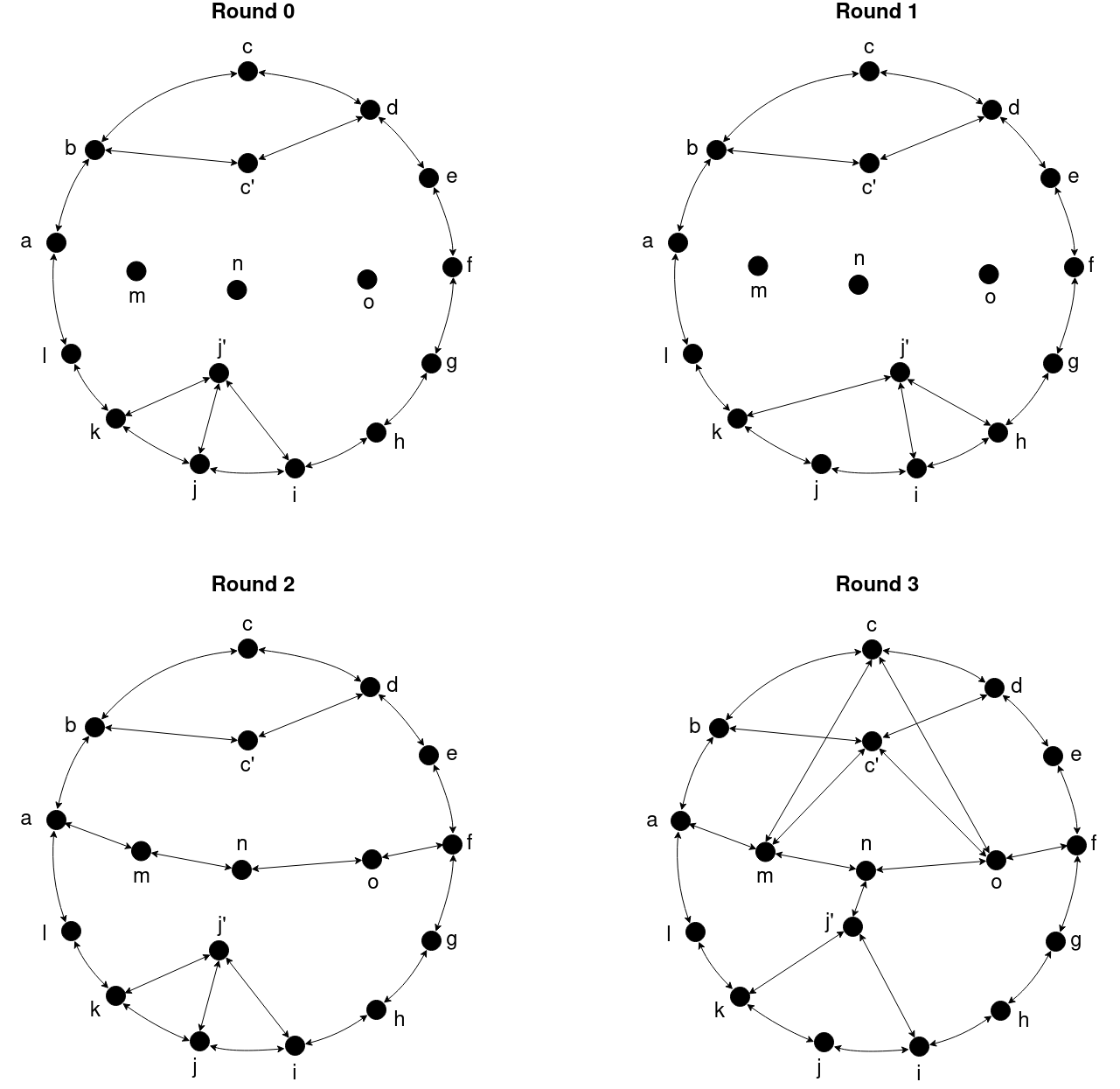}}
    \caption{A periodic time-varying graph with a period of 4, where nodes $c$ and $c'$ are $(4, 0)$-twins.}
    \label{fig_tvg}
\end{figure}

For $u\neq v$ two nodes of $\mathcal G$,
we say that $u$ and $v$ are $(\Delta, d)$-twins if, starting from a specific time instant $t_0 \in T$, their outside neighbourhoods have a non-empty intersection while the difference between these two sets is less than or equal to a fixed value $d$ for $\Delta$ consecutive rounds, 
that is, $|N_t(u) \setminus N_t(v) \setminus \{u, v\}| + |N_t(v) \setminus N_t(u) \setminus \{u, v\}| \leq d$, for every 
$t\in\llbracket t_0,t_0+\Delta-1\rrbracket$ if $t_0+\Delta\leq p$ and
for every $t\notin\rrbracket t_0+\Delta-1\mod p,t_0\llbracket$ 
otherwise.
Fig.~\ref{fig_tvg} exemplifies $(\Delta,d)$-twins.
We address the following problem.

\medskip

\noindent\textsc{Periodic$(\Delta,d)$-Twins}\\
\noindent\textsc{Input:} $\mathcal G=(G_t)_{t\in T}$ with $T=\llbracket0,p-1\rrbracket$ representing a time-varying graph of period $p$; integers $\Delta\leq p$ and $d$.\\
\noindent\textsc{Output:} for every node $u$ in $\mathcal G$, a list of all its $(\Delta,d)$-twins, that is, a list of $(v,t)$'s with $v$ being a node ID and $t$ the starting time instant where $u$ and $v$ become $d$-twins for $\Delta$ instants.

\medskip

Note that it is required that the period $p$ be given among the input of every node, but not the graph $\mathcal G$.
During each round $t$, every node has a local knowledge of $\mathcal G$, which is the size of its current neighbourhood, in $G_t$.
In the next Section~\ref{sec:disco}, we prove that \textsc{Periodic$(\Delta,d)$-Twins} can be solved after $2p$ rounds, using messages of size $O(\delta_{\mathcal G}\log n)$, where $n$ is the total number of nodes and $\delta_{\mathcal G}\leq n$ is the maximum degree of the graphs $G_t$'s.

\section{Distributed computation of $(\Delta,d)$-twins}
\label{sec:disco}
The main idea of our solution is based on the following property of static graphs, which hints that every node can compare its $1$-neighbourhood with the $1$-neighbourhood of any other node in its $2$-neighbourhood in order to determine whether that node is one of its $d$-twins.
For $u\neq v$ two nodes of a static graph $G=(V,E)$, we say that $u$ and $v$ are
$d$-twins if they have at least one common neighbour and that the difference between the sets of their outside neighbours is at most $d$, namely $|N(u) \setminus N(v) \setminus \{u, v\}| + |N(v) \setminus N(u) \setminus \{u, v\}| \leq d$.
A path of length $2$ between $u$ and $v$ is a path having $2$ edges and $3$ nodes, of the form $(u,w,v)$ with $uw\in E$ and $vw\in E$, for some $w\neq u$ and $w\neq v$.
\begin{prop}\label{prop:kpaths}
Let $G=(V,E)$ be an undirected static graph, let $u\neq v$ be two distinct nodes in $V$ with a common neighbour.
Then, $u$ and $v$ are $d$-twins if and only if $d\geq k$ or there are at least $k-d$ paths of length $2$ between $u$ and $v$, where $k = |N(u) \cup N(v) \setminus \{u,v\}|$.
\end{prop}
\begin{proof}
    The case when $d\geq k$ is clear.
    Otherwise let $A=N(u)\setminus\{u,v\}$ and $B=N(v)\setminus\{u,v\}$.
    Then, $k=|A\cup B|$.
    The paths of length $2$ between $u$ and $v$ are only paths passing through their common neighbours.
    Therefore, the number of such paths is $n_p=|A \cap B|$.
    We need to prove that $u$ and $v$ are $d$-twins if and only if $n_p\geq k-d$.
    From the definition of $d$-twins, we have $|A\setminus B|+|B\setminus A|\leq d$.
    Since $|A\cup B|-|A\cap B|=|A\setminus B|+|B\setminus A|$, this is equivalent to $|A\cup B|-|A\cap B|\leq d$,
    in other words, $k-n_p\leq d$.
    \qed
\end{proof}
\begin{remark}
Note that $u$ and $v$ could have very different neighbourhoods.
For instance when they have no common neighbours, $k$ could be large, while $u$ and $v$ are not $d$-twins until $d\geq k$.
\end{remark}

Let $\mathcal G=(G_t)_{t\in T}$ be a $p$-periodic time-varying graph with $T=\llbracket0,p-1\rrbracket$.
For every node to compute its list of $d$-twins, we would need $2$ rounds per static graph $G_t$, so that the node collects information about its $2$-neighbourhood.
Since the time-varying graph is periodic, we can proceed a first phase of $p$ rounds for passing every node's $1$-neighbourhood information to its neighbours.
Then, in a second phase consisting of $p$ extra rounds, every node can forward all the pieces of information it has already collected from its neighbours during the first phase.
The information received in this second phase informs every node about its $2$-neighbourhood.
There are however two drawbacks of the above idea.

Firstly, we need to control the size of the messages sent from each node, especially in the second phase.
In Lemma~\ref{lemma:d-twins} below, we prove that the information related to the entire neighbourhood is not necessary, but only the size of it.
Roughly, we plan to only send the size of the neighbourhood each time we need to forward this information.
More specifically, following Lemma~\ref{lemma:d-twins} each node $u$ in $\mathcal G$ can determine whether a node $v$ in $\mathcal G$ is such that $u$ and $v$ are $d$-twins at round $t$, by adding the number of $u$'s neighbours, obtained at round $t$, to the number of $v$'s neighbours,
obtained at round $t'=t+p$ and subtracting the number of common neighbours between them.
Then, if the result is less than $d$, they are $d$-twins, according to the inequality given in Lemma~\ref{lemma:d-twins}.

\begin{algorithm}[t!]
    \caption{Finding the $(\Delta,d)$-twins, all tables are dictionary data structures}\label{alg:main}
    \begin{algorithmic}[1]
    \Procedure{Periodic$(\Delta,d)$-Twins}{$p,\Delta,d$}
    
    \Comment{Initialisation}
        \State Table \texttt{Tab} with $p$ associating to each round the set of neighbours and the number of their neighbours
        \State Table \texttt{Count} associating to a given ID, the number of times it has been received in the current round
        \State $nbMsg \gets 0$ \Comment{The number of messages received.}
        \State Table \texttt{Neighbors} associating to a given ID its number of neighbours in the current round
        \State Table \texttt{Twins} storing the consecutive $d$-twins of the node, updated at each round, initially empty
        \State Table \texttt{dTwinsRound} associating to each round the found $d$-twins of the node
        \State Set \texttt{Delta\_d\_Twins} initially empty

    \For{each round}
        \State Call Algorithm~\ref{alg:roundI} \Comment{Sending a message}
        \State Call Algorithm~\ref{alg:receivingMsgDeltaTwins} \Comment{For received messages}
        
    \EndFor
    \If{$round = 2p-1$}
        \State $i \gets 0$
        \While{$i < p$ and $\texttt{Twins} \neq \emptyset$}
            \For{$id$ in \texttt{Twins}}
                \If{$id$ in \texttt{dTwinsRound[$i$]}}
                    \State $\texttt{Twins[$id$]} \gets \texttt{Twins[$id$]}+1$

                    \If{$\texttt{Twins}[id] = \Delta$}
                        \State Append $\{id, p + i - \Delta +1\}$ to \texttt{Delta\_d\_Twins}  

                        \State Remove $id$ from \texttt{Twins}
                            
                    \EndIf
                \Else
                    \State Remove $id$ from \texttt{Twins}
                \EndIf
            \EndFor
            \State $i \gets i + 1$
        \EndWhile
    
        \State \Return \texttt{Delta\_d\_Twins}
    \EndIf
    \EndProcedure
    \end{algorithmic}
\end{algorithm}
Secondly, we aim at computing $(\Delta,d)$-twins in $\mathcal G$, with $\Delta$ representing a continuous window of $\Delta$ consecutive time instants starting at $t_0\in T$.
For most $t_0$ this can be done using in the receiver node's internal storage a counter per sender node, which increments every time the receiver node detects a sender node as its $d$-twin at some $G_t$.
The counter is reset to zero whenever that sender node is detected as a non-$d$-twin at some (other) $G_t$.
However, cases involving extremities of interval $T$ such as when $t_0+\Delta\geq p$ will have their time window divided into two disjoint intervals: $\llbracket0,t_0+\Delta-1\mod p\rrbracket$ and $\llbracket t_0,p-1\rrbracket$.
If we proceed as with all other cases, the information will be ready at some round after the $2p$-th round for this case.
We can accelerate this to be as early as in the $2p$-th round using internal storage of a table in each receiver node, instead of the previously mentioned incremental counters.
Another positive consequence of using a table instead of using incremental counters is that, starting from round $p+\Delta$, every node has access in real time to its $(\Delta,d)$-twins of previous $\Delta$ rounds.
Then, at the $2p$-th round, every node has access to the full list of its $(\Delta,d)$-twins, where the computation terminates.

Algorithm~\ref{alg:main} implements the above ideas, calling as subroutines both Algorithms~\ref{alg:roundI}~and~\ref{alg:receivingMsgDeltaTwins} at every round.
All tables in Algorithm~\ref{alg:main} are dictionary data structures.
The result for every node $u$ will be stored in table \texttt{Delta\_d\_Twins}.
For its computation, node $u$ maintains furthermore two tables, called \texttt{Twins} and \texttt{Count}.
For a node $v\neq u$ with identifier $id$, a strictly positive value of \texttt{Twins[$id$]} means that $v$ and $u$ have been twin vertices for the last \texttt{Twins[$id$]} number of rounds.
Hence, when \texttt{Twins[$id$]} is at least $\Delta$, we append $id$ to the resulting list \texttt{Delta\_d\_Twins} along with the number of rounds where they start to be $(\Delta,d)$-twins.
This is implemented in Algorithm~\ref{alg:receivingMsgDeltaTwins}, lines 21-22.
\begin{algorithm}[t!]
    \caption{Send messages}\label{alg:roundI}
    \begin{algorithmic}[1]
    \Procedure{SendMessage}{$p,\Delta,d,round$}
        \If{$round < p$}
            \State $msg \gets <myId,\ nbNeighbors>$
            \State Send $msg$ to all neighbours.
        \ElsIf{$p\le round < 2*p$}
            \State $msg \gets \emptyset$ 
            \State $index \gets round-p$
            \For{$elem \in \texttt{Tab[$index$]}$}
                \State Append $elem$ to $msg$
            \EndFor
            \State Send $msg$ to all neighbours
        \EndIf
    \EndProcedure
    \end{algorithmic}
\end{algorithm}
\begin{algorithm}[t!]
    \caption{Upon Receiving a Message \textbf{msg}}\label{alg:receivingMsgDeltaTwins}
    \begin{algorithmic}[1]
    \Procedure{ReceiveMessage}{$p,\Delta,d,round$}
        \If{$round < p$}
            \State Append $msg$ to $\texttt{Tab[$round$]}$
        \ElsIf{$p\le round < 2*p$}
            \State $nbMsg \gets nbMsg + 1$
            \For{$<id, nb>$ in $msg$}
                \State $\texttt{Count[$id$]} \gets \texttt{Count[$id$]} + 1$
                \State $\texttt{Neighbors[$id$]} \gets nb$
                
            \EndFor
            
            \If{$nbMsg = nbNeighbors$}

                \For{each $id$ in \texttt{Twins}}
                    \If{$id$ not in \texttt{Count}}
                        \State Remove $id$ from \texttt{Twins}
                    \EndIf
                \EndFor
                
                \For{each $id$ in \texttt{Count}}
                    
                    \State $dTwin \gets nbNeighbors + \texttt{Neighbors[$id$]} - 2*\texttt{Count[$id$]}$
                    
                    \If{$dTwin \le d$}
                        
                        \State Append $id$ to $\texttt{dTwinsRound[$round - p$]}$
                        \If{$id$ in \texttt{Twins}}
                            \State $\texttt{Twins[$id$]} \gets \texttt{Twins[$id$]}+1$
                        \Else
                            \State $\texttt{Twins[$id$]} \gets 1$
                        \EndIf

                        \If{$\texttt{Twins[$id$]} \geq \Delta$}
                            \State Append $\{id, round-\Delta+1 \}$ to \texttt{Delta\_d\_Twins}  

                        \EndIf
                    \Else
                        \State Remove $id$ from \texttt{Twins}
                    \EndIf
                \EndFor
            \EndIf
        \EndIf

    \EndProcedure
    \end{algorithmic}

\end{algorithm}

In case we conclude that $id$ is not a twin with $u$, such as with Algorithm~\ref{alg:receivingMsgDeltaTwins}'s lines 12 and 24,
a quick way to save this information is to remove key $id$ from the (dictionary) table \texttt{Twins}.
Table \texttt{Count} is to be used in Algorithm~\ref{alg:receivingMsgDeltaTwins} and related to the current round.
We use an alternative computation for what is stored in table \texttt{Count}: rather than using Property~\ref{prop:kpaths} which would force us to count the number of paths of length $2$, we use the equivalent quantities showed in Lemma~\ref{lemma:d-twins} below instead.
These quantities from Lemma~\ref{lemma:d-twins} are encoded precisely at line 14 of Algorithm~\ref{alg:receivingMsgDeltaTwins}.
The size of the message in the second phase is composed by two additive terms.
The first term is the number of neighbours, multiplied by the size of the nodes ID's.
The second term is the size of the neighbourhood of every neighbour, which is upper bounded by the first term.
Thus, the maximum size of messages are bounded by $O(\delta_{\mathcal G}\log n)$, where $n$ is the total number of nodes and $\delta_{\mathcal G}\leq n$ is the maximum degree of the graphs $G_t$'s, \textit{cf.} Lemma~\ref{lemma:sizeOfMsg} below.
Fig.~\ref{fig_tvgmsgs} exemplifies the messages received in the second phase of every node in the example given in Fig.~\ref{fig_tvg}.

\begin{figure}[t!]
    \centerline{\includegraphics[scale=0.2]{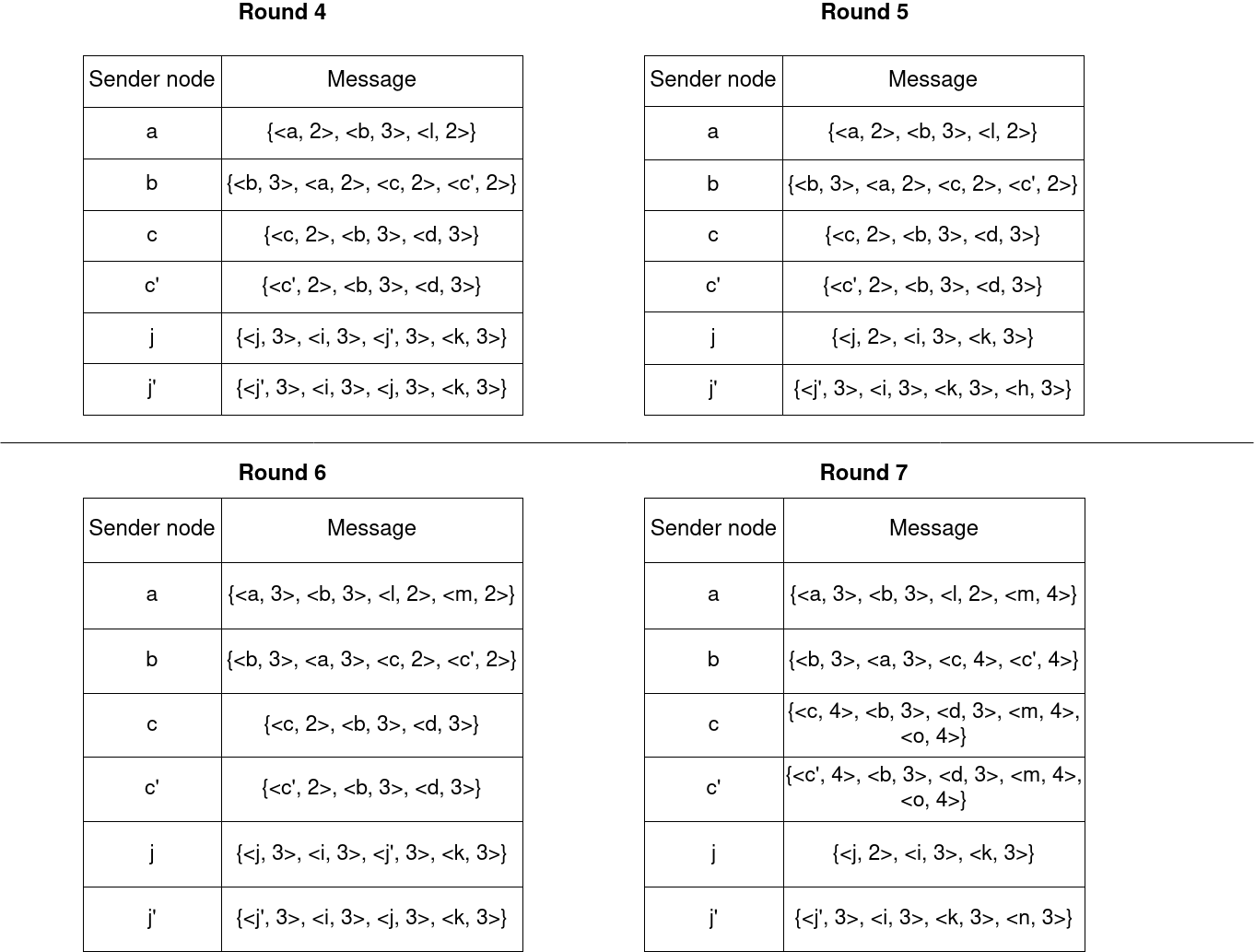}}
    \caption{Tables describing the messages sent to the direct neighbourhood of nodes in the example presented in Fig.~\ref{fig_tvg}. The messages described here-above are sent through the second period of the time-varying graph, corresponding to the information collected during the first period.}
    \label{fig_tvgmsgs}
\end{figure}

We now prove the correctness of the algorithm, as well as the maximum size of the messages used in the algorithm.
\begin{lemma}\label{lemma:sizeOfMsg}
    Given a p-periodic time-varying graph $\mathcal G=(G_t)_{t\in T}$, with $n$ nodes, if each node sends the list of IDs of its neighbours associated with the number of their neighbours collected at a round $t$ in Algorithm~\ref{alg:receivingMsgDeltaTwins} in line 3, then the maximum size of a message sent at round $t+p$ as stated in Algorithm~\ref{alg:roundI} lines 6-11 is $O(\delta_{\mathcal G}\log(n\delta_{\mathcal G}))$.
\end{lemma}

\begin{proof}
  A node $u$ can have at most $\delta_{\mathcal G} = n-1$ neighbours. For every neighbour $w$ of $u$, $u$ needs to store the ID of $w$, which requires a size of $O(\log n)$, and the number of neighbours of $w$, which can be at most $\delta_{\mathcal G}$, requiring a size $\log \delta_{\mathcal G}$. So for a neighbour $w$ of $u$, $u$ stores information of size $O(\log n) + \log \delta_{\mathcal G} = O(\log n\delta_{\mathcal G})$, and for all neighbours of $u$, it stores information of size $O(\delta_{\mathcal G} \log n\delta_{\mathcal G})$, and this information represents the size of the message that $u$ will send in the second phase.
  \qed
\end{proof}

\begin{lemma}\label{lemma:d-twins}
    Given a p-periodic time-varying graph $\mathcal G=(G_t)_{t\in T}$, and two distinct nodes $u\neq v$,
    $u$ and $v$ are $(\Delta,d)$-twins if
    $|N_t(u) \setminus \{u, v\}| + |N_t(v) \setminus \{u, v\}| - 2\times|((N_t(u) \setminus \{u, v\}) \cap (N_t(v) \setminus \{u, v\}))| \le d$ for $t_0\leq t<t_0+\Delta$ that is, for $\Delta$ consecutive rounds.
\end{lemma}
\begin{proof}
The proof is very similar to that of Property~\ref{prop:kpaths}.
    Let $A=N_t(u)\setminus\{u,v\}$ and $B=N_t(v)\setminus\{u,v\}$.
    From the definition of $d$-twins, we have $|A\setminus B|+|B\setminus A|\leq d$.
    Since $|A|+|B|-2\times|A\cap B|=|A\setminus B|+|B\setminus A|$, this is equivalent to $|A|+|B|-2\times|A\cap B|\leq d$.
    \qed
\end{proof}

The latter lemma allows us to store the $d$-twins in Algorithm~\ref{alg:receivingMsgDeltaTwins} and check the number of consecutive round two nodes are $d$-twins in Algorithm~\ref{alg:receivingMsgDeltaTwins} line 26 and Algorithm~\ref{alg:main} line 20.
We have proved the following theorem.
\begin{theorem}\label{theo:main}
In a $p$-periodic time-varying graph $\mathcal G=(G_t)_{t\in T}$ where integers $p,\Delta,d$ are given as input to every node, problem \textsc{Periodic$(\Delta,d)$-Twins} can be solved
after $2p$ rounds, using messages of size $O(\delta_{\mathcal G}\log n)$, where $n$ is the total number of nodes and $\delta_{\mathcal G}\leq n$ is the maximum degree of the graphs $G_t$'s.
\end{theorem}

\section{Twin sampling with $O(\log n)$ message size}
\label{sec:hash}
In a static graph, we can reduce the message size as follows.
According to Lemma~2\ in~\cite{HNT22}, every node $v$ can receive its neighbourhood in a first round, then apply a well selected hash function sampled from a universe $\mathcal U$ of hash functions and forward this to every neighbour $u$ in a second round.
When receiving the second round message, node $u$ can compute the value of $|N(u)\cap N(v)|$ within $\epsilon\max(n_u,n_v)$ with probability $1-\nu$, where $n_u=|N(u)|$ and $n_v=|N(v)|$.
The process uses $O(1)$ messages of size $O(\epsilon^4\log(1/\nu)+\log\log|\mathcal U|+\log\max(n_u,n_v))$ bits.
Whence, w.h.p.\ after $2$ rounds the inequality in our Lemma~\ref{lemma:d-twins} can be decided using messages of $O(\log n)$ bits.
For the dynamic case, the extension is similar to construction proposed in Section~\ref{sec:disco}.

\section{Conclusion and perspectives}
\label{sec:conclu}
We introduce the problem of finding $(\Delta,d)$-twins  and propose a distributed algorithm to compute them in any $p$-periodic time-varying graph $\mathcal G=(G_t)_{t\in\mathbb N}$ under a distributed model similar to the $1$-interval-connected network.
After $2p$ rounds, every node can compute the nodes that are its $(\Delta,d)$-twins using messages of size $O(\delta_\mathcal G\log n)$, where $n$ is the total number of nodes and $\delta_\mathcal G$ is the maximum degree of the graphs $G_t$'s.
Using techniques borrowed from~\cite{HNT22}, we reduce the message size down to $O(\log n)$ w.h.p.
Finding $(\Delta,d)$-twins can be useful in several ways.
For instance, it could be used for alternately scheduling sleeping times of the twin nodes to save resources, while maintaining connectivity for the rest of the network. Furthermore, it can be used in order to compute disjoint paths or disjoint broadcast trees.
As for the next steps of research, it would be useful to extend our algorithm to compute $\epsilon$-modules as defined in~\cite{HMZ20}.


\bibliographystyle{splncs04}
\bibliography{references}
\end{document}